\documentclass[11p, letter]{article}
\usepackage{amsmath,amsfonts}
\usepackage{gastex}
\usepackage{dcolumn}
\usepackage{multicol}
\usepackage{color}
\usepackage{dsfont}
\usepackage{xspace}
\usepackage{epsfig}
\usepackage{multirow}
\usepackage{lscape}
\usepackage{array}
\usepackage{dcolumn}
\usepackage{rotating}
\usepackage{booktabs} 
\usepackage{varwidth}
\usepackage{fullpage}
\usepackage{float}
\newcommand{\maxmult}{\ensuremath{\mathit{maximal\_multiplicity}}\xspace}
\newcommand{\observe}{\ensuremath{\mathit{observe\_neighbors}}\xspace}
\newcommand{\E}{\mathbb E}
\newcommand{\PP}{\mathbb P}
\newcommand{\NN}{\mathbb N}

\newcommand{\ZZ}{\mathbb Z}

\newcommand{\qedsymb}{\hfill{\rule{2mm}{2mm}}}
\newenvironment{proof}
{\begin{trivlist}
\item[\hspace{\labelsep}{\bf\noindent Proof: }]
}
{\qedsymb\end{trivlist}}
\newtheorem{lemma}{Lemma}

\floatstyle{ruled}
\newfloat{algorithm}{thp}{lop}[section]
\floatname{algorithm}{Algorithm}

\newtheorem{definition}{Definition}
\newtheorem{remark}{Remark}
\newtheorem{note}{Note}
\newtheorem{texample}{Toy example}
\newcommand{\remove}[1]{}
\begin{document}

\title{The Cost of Probabilistic Agreement in Oblivious Robot Networks}

\author{Julien Clement$^{\star\star}$\\
{\footnotesize clement@lri.fr}
\and Xavier Defago$^{\circ}$\\
{\footnotesize defago@jaist.ac.jp}
\and Maria Gradinariu Potop-Butucaru$^\star$\\
{\footnotesize maria.gradinariu@lip6.fr}
\and Stephane Messika$^{\star\star}$\\
{\footnotesize messika@lri.fr}\\
{\footnotesize
\begin{tabular}{c c c}
$^\star$ Universit\'{e} Pierre et Marie Curie - Paris 6&
$^{\star\star}$ Universit\'{e} Paris-Sud - Paris 11 & $^\circ$ School of
Information Science\\
LIP6/CNRS UMR 7606 &
LRI/CNRS UMR 8623 & JAIST \\
104 avenue du Pr\'{e}sident Kennedy &
Laboratoire de Recherche en Informatique \\
75016 Paris, France &
91405 Orsay, France & Ishikawa, Japan\\
\end{tabular}}
}

\date{}

\remove{
\author{
\alignauthor Julien Clement\titlenote{Work supported by R\'egion \^ile-de-France}\\
       \affaddr{L.R.I}\\
       \affaddr{Universit\'e Paris Sud}\\
       \affaddr{Orsay, France}\\
       \email{clement@lri.fr}
\alignauthor Xavier Defago\titlenote{Work supported by MEXT Grant-in-Aid for Young Scientists (A) (Nr. 18680007).}\\
       \affaddr{School of information Science}\\
       \affaddr{JAIST}\\
       \affaddr{Ishikawa, Japan}\\
       \email{defago@jaist.ac.jp}
\alignauthor Maria Gradinariu Potop-Butucaru\\
       \affaddr{IRISA}\\
       \affaddr{Universit\'e Pierre et Marie Curie}\\
       \affaddr{Paris 6, France}\\
       \email{maria.gradinariu@lip6.fr}
\and  
\alignauthor Stephane Messika \\
       \affaddr{L.R.I}\\
       \affaddr{Universit\'e Paris Sud}\\
       \affaddr{Orsay, France}\\
       \email{messika@lri.fr}
 }
}       
 \maketitle

\begin{abstract} 
  In this paper we address the complexity issues of two 
agreement problems in oblivious robot networks namely gathering and
scattering. These abstractions are fundamental coordination problems in cooperative
  mobile robotics.  Moreover, their oblivious characteristics makes them appealing for self-stabilization 
since they are self-stabilizing with no extra-cost. 
Given a set of robots with arbitrary initial location and no initial agreement on a global coordinate
  system, {\it gathering} requires that all robots reach the exact same 
but not predetermined location while {\it scattering} aims at scatter robots such that no two robots share the
same location. Both deterministic gathering and scattering have been 
proved impossible under arbitrary schedulers therefore 
probabilistic solutions have been recently proposed. 

The contribution of this paper is twofold. First, we propose a
detailed complexity analysis of the existent probabilistic 
gathering algorithms  in both fault-free and
fault-prone environments\footnote{We consider both crash and byzantine-prone environments}. 
Moreover, using Markov chains tools and additional
assumptions on the environment we prove that the gathering 
convergence time can be reduced from $O(n^2)$ (the best known tight 
bound) to $O(nln(n))$. Additionally, we prove that in crash-prone environments 
gathering is achieved in $O(nln(n)+2f)$. 
Second, using the same technique we prove that the best known scattering strategy
converges in fault-free systems is $O(n)$ (which is one to optimal) 
while  in crash-prone environments it needs $O(n-f)$.
Finally, we conclude the paper with a discussion related 
to different strategies to gather oblivious robots.
\end{abstract}

\section{Introduction}
\bigskip
Many applications of mobile robotics envision groups of mobile robots
self-organizing and cooperating toward the resolution of common
objectives. In many cases, the group of robots is aimed at being
deployed in adverse environments, such as space, deep sea, or after
some natural (or unnatural) disaster. It results that the group must
self-organize in the absence of any prior infrastructure (e.g., no global
positioning), and ensure coordination in spite of faulty robots and
unanticipated changes in the environment.

Suzuki and Yamashita \cite{SY99} proposed a formal model to analyze
and prove the correctness of agreement problems in robot networks. 
In this model, robots are represented as
points that evolve on a plane. At any given time, a robot can be
either idle or active. In the latter case, the robot observes the
locations of the other robots, computes a target position, and moves
toward it. The time when a robot becomes active is governed by an
activation daemon (scheduler). Between two activations robots forget
the past computations. Interestingly, any algorithm proved correct 
in this model is also self-stabilizing. 

The \emph{gathering problem}, also known as the
\emph{Rendez-Vous} problem, is a fundamental coordination problem in
oblivious mobile robotics.  In short, given a set of robots with
arbitrary initial location and no initial agreement on a global
coordinate system, gathering requires that all robots, following their
algorithm, reach the exact same location---one not agreed upon
initially---within a \emph{finite} number of steps, and remain there. 
The dual of the gathering problem is the \emph{scattering problem}. 
Started in a arbitrary configuration scattering requires that
eventually no two robots share the same position.
Both scattering and gathering are agreement problems and
similar to the consensus problem in conventional distributed systems,
they have simple definitions but the existence of a solution
greatly depends on the synchrony of the systems as well as the nature
of the faults that may possibly occur. The task becomes even harder
since robots are anonymous since the impossibility results proved in classical distributed computing 
also hold in robot networks. Therefore,
specific problems like flocking, gathering or scattering are impossible without additional assumption. 
Interestingly, most of the work done so far in order to convey the
above impossibility results focuses on the additional assumptions the system
needs, less attention being shown 
to the use of randomization. Surprisingly, no formal framework 
was proposed in order to analyze the correctness and the
complexity of probabilistic algorithms designed for robots networks. 
In a companion paper, \cite{DGMP-disc06}, we investigated
some of the fundamental limits of deterministic and probabilistic
gathering face to a broad synchrony and fault
assumptions. Probabilistic scattering,
was analyzed for the first time in \cite{code}. None
of the previously mentioned works focus on a framework in
order to compute the complexity of proposed solutions.

In this paper we advocate that Markov chains are a simple and
efficient tool to analyze and compare 
probabilistic strategies in both fault-free and fault-prone
oblivious robot networks. Note that in robot networks computations
depend only on the current view of the robots without the use of 
the past. This behavior makes Markov chains an appealing
tool for analyzing their correctness and complexity since 
by definition a Markov chain models systems where 
the next configuration depends strictly on the current configuration. 
The only difficulty while using Markov chains 
is to associate to each probabilistic strategy the
appropriate Markov chain. In the current work we focus on the analysis of
existing probabilistic strategies for scattering and gathering
and advocate that our analysis can be easily applied to a broad class
of probabilistic strategies (e.g. leader election, 
flocking, constrained scattering, pattern formation).


\paragraph*{Contribution}
Our contribution is twofold. First, we show that the time complexity
of probabilistic gathering in
fault-free environments can be
improved from $O(n^2)$ to $O(nln(n))$ when the algorithms exploit
additional information related to the environment (eg. multiplicity 
knowledge). Additionally, in crash-prone
environments we prove that the convergence time of gathering is $O(nln(n)+2f)$. Second, 
we show that the tight bound for 
scattering is $O(n)$ (which is one to optimal) in fault-free systems
while in crash-prone environments
scattering converges is $O(n-f)$ rounds where
$f$ is the maximal bound on the number of faults. 
Intuitively gathering and scattering should have the same complexity however 
our analysis shows that scattering is much easier to obtain than gathering, $O(n)$ versus $O(nln(n))$. 
Reducing the gap in complexity between gathering and scattering seems to be 
an interesting research direction as it will be discussed 
in Section \ref{sec:gathering_via_scattering}.

\paragraph*{Structure of the paper}
The  paper is structured as follows. Section \ref{sec:Model} describes the robots network and system model. 
Section  \ref{sec:Problem} formally defines the gathering and
scattering problems. 
We propose the complexity analysis of existent probabilistic
scattering and gathering in Sections \ref{sec:gathering} and \ref{sec:scattering} in fault-free
and fault-prone environments. In Section
\ref{sec:gathering_via_scattering} we analyze an alternative strategy to 
gather oblivious robots. Section \ref{sec:conclusion} concludes the paper and
discusses some open problems.

\section{Model}
\label{sec:Model}

In the following we propose the model of our system. Most of the
definitions are borrowed from \cite{SY99,Pre01}.

\paragraph{Robot networks.}
We consider a network of a finite set of
robots arbitrarily deployed in a geographical area. The robots are
devices with sensing, computational and motion capabilities. They can
observe (sense) the positions of other robots in the plane and based
on these observations they perform some local
computations. Furthermore, based on the local computations robots may
move to other locations in the plane.

In the case robots are 
able to sense the whole set of robots they are referred 
as robots with \emph{unlimited visibility}; otherwise 
robots have limited visibility. In this paper, we consider that 
robots have unlimited visibility. 

In the case robots are able to distinguish if 
there are more than one robot at a given position they are referred 
as robots with \emph{multiplicity knowledge}.

\paragraph{System model.}
A network of robots that exhibit a discrete behavior
can be modeled with an I/O automaton~\cite{Lyn96}. A network of robots that
exhibit a continuous behavior can be modeled with a hybrid I/O automaton~\cite{LSV03}.
This framework allows the modeling of systems that exhibit both a discrete
and continuous behavior and in particular the modeling of robots
networks.
 
The actions performed by the automaton modeling a robot are as follows:
\begin{itemize}
\item \emph{Observation (input type action)}.\\
  An observation returns a snapshot of the positions of all the robots
  in the visibility range.
  In our case, this observation returns a snapshot of the positions of
  all the robots;
\item \emph{Local computation (internal action)}.\\
  The aim of a local computation is the computation of a destination
  point;
\item \emph{Motion (output type action)}.\\
  This action commands the motion of robots towards the destination
  location computed in the local computation action.
\end{itemize}
 
The local state of a robot at time~$t$ is the state of its
input/output variables and the state of its local variables and
registers.  A network of robots is modeled by the parallel composition
of the individual automata that model one per one the robots in the
network.  A configuration of the system at time~$t$ is the union of
the local states of the robots in the system at time~$t$.  An
execution $e=(c_0, \ldots, c_t, \ldots)$ of the system is an infinite
sequence of configurations, where $c_0$ is the initial configuration%
\footnote{Unless stated otherwise, this paper makes no specific
  assumption regarding the respective positions of robots in initial
  configurations.}
of the system, and every transition $c_i \rightarrow c_{i+1}$ is
associated to the execution of a subset of the previously defined
actions.

\paragraph{Schedulers.} A scheduler decides at each configuration the
set of robots allowed to perform their actions.
A scheduler is fair if, in an infinite execution, a robot is activated
infinitely often.  In this paper we consider the fair version of the
following schedulers:
\begin{itemize}
\item \emph{centralized}: at each configuration a single robot is
  allowed to perform its actions;
\item \emph{probabilistic}: at each configuration a set 
of robots chosen uniformly is activated;
\item \emph{$k$-bounded}: between two consecutive activations of a
  robot, another robot can be activated at most $k$~times;
\item \emph{arbitrary}: at each configuration an arbitrary subset of
  robots is activated.
\end{itemize}

\paragraph{Faults.}
In this paper, we address the following failures:
\begin{itemize}
\item \emph{crash failures}: In this class, we further distinguish two
  subclasses: (1)~robots physically disappear from the network, and
  (2)~robots stop all their activities, but remain physically present
  in the network;
\item \emph{Byzantine failures}: In this case, robots may have an
  arbitrary behavior.
\end{itemize}

\paragraph{Computational models.}
The literature proposes two computational models: ATOM and CORDA.  The
ATOM model was introduced by Suzuki and Yamashita~\cite{SY99}. In this
model each robot performs, once activated by the scheduler, a
\emph{computation cycle} composed of the following three actions:
observation, computation and motion. The particularity 
of the ATOM model versus the CORDA model is that in ATOM a computation
cycle is atomic while in CORDA the atomicity concerns only the
actions. In order words CORDA models and asynchronous networks 
while ATOM a synchronous one.

In this paper, we consider the ATOM model. Moreover, we consider that robots are
oblivious (i.e.,~stateless). That is, robots do not conserve any
information between two computational cycles.%
\footnote{One of the major motivation for considering oblivious robots
  is that, as observed by Suzuki and Yamashita~\cite{SY99}, any
  algorithm designed for that model is inherently self-stabilizing.}
We also assume that all the robots in the system have unlimited
visibility.

\section{Gathering and Scattering}
\label{sec:Problem}
A network of robots is in a \emph{legitimate configuration} with respect to 
the gathering requirement if all robots in the system share the same position in 
the plane. 
Let denote by ${\mathcal P}_\mathit{Gathering}$ this predicate.

An algorithm solves the gathering problem in an oblivious system 
if the following two properties are verified:
\begin{itemize}
\item \textbf{Convergence} Any execution of the system starting in an arbitrary 
configuration reaches in a finite number of steps a configuration that satisfies 
$\mathcal{P}_{Gathering}$.
\item \textbf{Closure} Any execution starting in a legitimate 
configuration with respect to the $\mathcal{P}_\mathit{Gathering}$ predicate 
contains only legitimate configurations.
\end{itemize}

Gathering is difficult to achieve in most of the environments. Therefore,
weaker forms of gathering were studied so far. An interesting 
version of this problem requires robots to
\emph{converge} toward a single location rather than reach that
location in a finite time. The convergence is however
considerably easier to deal with. For instance, with unlimited visibility,
convergence can be achieved trivially by having robots moving toward
the barycenter of the network \cite{SY99}.  


Scattering, introduced first in \cite{suzuki-yamashita}, aims at arranging a set of
robots such that eventually no two robots share the same position. Let
denote by ${\mathcal P}_{Scattering}$ this predicate.
Formally, scattering is defined by the following two properties :
\begin{itemize}
\item \textbf{Convergence} Any execution of the system starting in an arbitrary 
configuration reaches in a finite number of steps a configuration that satisfies 
$\mathcal{P}_{Scattering}$. 
\item \textbf{Closure} Any execution starting in a legitimate configuration with
respect to the predicate ${\mathcal P}_{Scattering}$ contains only
legitimate configurations.
\end{itemize}

In the sequel we address the convergence time of gathering and scattering in both 
fault-free and fault-prone environments. As stated in the model 
we consider both crash-prone and byzantine-prone systems.
Let $(n,f)$ denote a system with $n$ correct robots but $f$ and the considered faults are 
crashes and byzantine behavior.  As mentioned in Section~\ref{sec:Model} in
a $(n,f)$ crash-prone system there are two types of crashes: (1)~the
crashed robots completely disappear from the system, and (2)~the
crashed robots are still physically present in the system, however
they stop the execution of any action.
In \cite{DGMP-disc06}, we proved that if the faulty
robots disappear from the system, then the problem trivially reduces
to the study of its fault-free version with $n\!-\!f$~correct
robots. In contrast, in systems where faulty robots remain physically
present in the network after crashing, the problem is far from being
trivial.  A similar argument can be provided for the case of byzantine
behavior.
Obviously, gathering or scattering all the robots in the system including the faulty
ones is impossible since faulty robots may possibly have crashed at
different locations or collude as shown in \cite{DGMP-disc06}.
Therefore, we study the feasibility of weaker versions of
gathering and scattering, referred to as \emph{weak gathering}
respectively \emph{weak scattering}. The
$(n,f)$-\emph{weak} problem requires that, in a terminal
configuration, only the \emph{correct} robots must verify the specification.  

\section{Analysis framework}
In this section we introduce some notations and definitions that will
be further used in order to analyze the convergence time of  
the probabilistic gathering and scattering. A detailed description of 
the notions defined below can be found in \cite{norris}.

\paragraph{Random variables}
We denote $X_n$ a random variable. For instance, in our case it might be the number of groups of size $x$ 
after $n$ steps of the algorithm. 
We will study a discrete-time stochastic process, that is : a sequence $\{X_n\}_{n \geq 0}$ of random variables.

In the sequel we will use the following notations:
\begin{itemize}
 \item $\PP[X_n=x]$ the probability of the event $\{X_n=x\}$.
 \item $\E[X_n]$ the expectation of $X_n$ .
 \item The probability distribution of a random variable $X : k \mapsto \PP[X=k] $ for all $k$
 \item Conditional probability will be written $\PP[A \mid B]$, and will be read "the probability of A, given B".
\end{itemize}

\paragraph{Markov chains}
Markov chains are particular classes of stochastic processes. These
stochastic processes have the following fundamental property : 
the probabilistic dependence on the past is only related to the previous state.  
 
\begin{definition}
Let $\left(X_n \right)_{n \in \NN}$ be a discrete time stochastic
process with countable state space $E$. 
If for all integers $n \geq 0$ and all states $i_0,i_1,\hdots,i_{n-1},i,j$:
$$\PP[X_{n+1}=j \mid X_n=i,X_{n-1}=i_{n-1},\hdots,X_0=i_0]=\PP[X_{n+1}=j \mid X_n=i]$$
Whenever both sides of the above equality are well defined, this
stochastic process is called Markov chain. 
A Markov chain is homogeneous (HMC) if the right side is independent of $n$.
\end{definition}

In the following we propose an example of Markov chain.
\begin{texample}
\label{toy:1}
Consider a robot performing a random walk in a two-dimensional
space. The robot at position $(i,j)$ chooses with equal
probability as destination point a neighbor of $(i,j)$ in the space $\ZZ\times\ZZ$.  
Let $X_n$ be the robot position after $n$ steps: 
\begin{itemize}
\item $\PP[X_{n+1}=(i,j) \mid X_n=(i-1,j)]=1/4$
\item $\PP[X_{n+1}=(i,j) \mid X_n=(i,j-1)]=1/4$
\item $\PP[X_{n+1}=(i,j) \mid X_n=(i+1,j)]=1/4$
\item $\PP[X_{n+1}=(i,j) \mid X_n=(i,j+1)]=1/4$
\end{itemize}
\end{texample}
 
\begin{note} Note that 
Toy Example \ref{toy:1} is a Markov Chain since every position of
the robot is only dependent on the previous one.
\end{note}

In this paper we advocate that Markov chains are a simple verification tool, perfectly adapted 
to the analysis of distributed strategies in oblivious robot networks since in these networks 
the next move of a robot depends only on its current position.

\section{Probabilistic Gathering}
\label{sec:gathering}
In this section we analyze the complexity of probabilistic gathering in fault-free
and fault-prone environments. The algorithms analyzed in this section
were proposed in \cite{DGMP-disc06}.

\subsection{Gathering in fault-free environments}
In this section we prove that additional information on the environment drastically improves 
the time convergence of gathering. Using multiplicity knowledge, for example, we obtain a tight bound of $O(nln(n))$ which improves 
the best known bound of $O(n^2)$. 
In \cite{DGMP-disc06} we proposed a probabilistic algorithm that
solves the fault-free gathering in ATOM, under a special class 
of schedulers, known as $k$-bounded schedulers. 
A  robot, when chosen by the scheduler, selects randomly one of its neighbors and moves towards
its position with probability $\frac{1}{\delta}$ where $\delta$ is the
size of the robot's view. In the considered model the robots have
unlimited visibility and the value of $\delta$ is $n$. 
We proved that this strategy probabilistically solves $2$-gathering in the \textsl{ATOM} model 
under an arbitrary scheduler and converges in $2$ steps in
expectation.  
We also proved that it solves the
  $n$-gathering problem ($n \geq 3$), under a fair
  $k$-bounded scheduler without multiplicity knowledge and converges under fair
  bounded schedulers in $n^2$ rounds \footnote{A round is the shortest
fragment of an execution i which each process in the system executed
at least once its actions.}
  in expectation.

In the following we show that the $O(n^2)$ complexity bound can be
reduced to $O(nln(n))$ when robots use the multiplicity. The algorithm
that meets this bound was initially proposed in \cite{DGMP-disc06}
in order to cope with robots crash. The algorithm, shown as Algorithm
\ref{alg:ft-prob-gathering}, works as follows.
When a robot is chosen by the scheduler moves to a group with
maximal multiplicity. When several groups have the same maximal
multiplicity, then a robot member of such group 
tosses a coin to decide if it moves or holds the current position.
Interestingly, the multiplicity knowledge (used so far in order to
break the symmetry of the system) can also be used in order to
fasten gathering. 

\begin{algorithm}
\hspace*{3cm}\textbf{Functions}:\\
\hspace*{3cm}$\observe::$ returns the set of robots within the\\
\hspace*{3cm}vision range of robot $p$ (the set of $p$'s neighbors);\\
\hspace*{3cm}$\maxmult::$ returns the set of robots with the \\
\hspace*{3cm}maximal multiplicity;\\
\noindent
 \hspace*{3cm}\textbf{Actions}:\\
    \noindent
    \hspace*{3.3cm}
    ${\mathcal A}_{1}::\ true \longrightarrow$ \\
    \hspace*{4cm} $\mathcal{N}_p=\observe()$;\\
    \hspace*{4cm} if $p \in \maxmult(\mathcal{N}_p) \wedge \left|\maxmult\left(\mathcal{N}_p\right)\right| > 1$ then\\
    \hspace*{5cm} 
    with probability $\frac{1}{\left|\maxmult\left(\mathcal{N}_p\right)\right|}$do\\
    \hspace*{6cm} select a robot $q \in \maxmult(\mathcal{N}_p)$;\\
    \hspace*{6cm} move towards $q$;\\
    \hspace*{4cm} else\\
    \hspace*{5cm} select a robot $q \in \maxmult(\mathcal{N}_p)$;\\
    \hspace*{5cm} move towards $q$;\\
\caption{Probabilistic gathering for robot $p$ with multiplicity knowledge.  \label{alg:ft-prob-gathering}}
\end{algorithm}

\begin{lemma}
\label{lemma:ff-ct}
In a fault-free environment the convergence time of Algorithm
\ref{alg:ft-prob-gathering} 
is $\alpha_nln(\alpha_n)+1$, with $\alpha_n=[\frac{n}{2}]+1$.
\end{lemma}

\begin{proof}
In order to study the convergence time of Algorithm
\ref{alg:ft-prob-gathering} we introduce the following stochastic process :
$\forall \ t , ~\ X_t=k$ means that at round $t$ the group with the
maximal multiplicity has $k$ robots. Note that when $k$ equals $1$ all robots are scattered
such that no two robots share the same position.
Figure \ref{fig:stochastic-ftprob} proposes the probability transition of the 
random variable $X_t$.

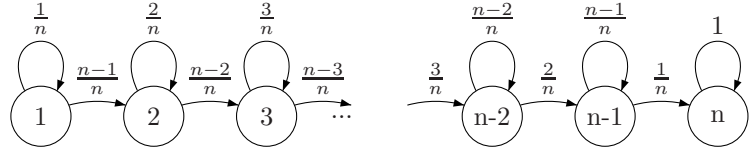
\begin{figure}

\hrule
\bigskip
\begin{picture}(130,25)
\node(A)(40,10){1}
\node(B)(55,10){2}
\node(C)(70,10){3}
\node(D)(100,10){n-2}
\node(E)(115,10){n-1}
\node(F)(130,10){n}
\node[Nframe=n](C')(85,10){}
\node[Nframe=n](D')(85,10){}
\node[Nframe=n](D'')(80,10){...}
\gasset{curvedepth=2}
\drawedge(A,B){$\frac{n-1}{n}$}
\drawedge(B,C){$\frac{n-2}{n}$}
\drawedge(C,C'){$\frac{n-3}{n}$}
\drawedge(D',D){$\frac{3}{n}$}
\drawedge(D,E){$\frac{2}{n}$}
\drawedge(E,F){$\frac{1}{n}$}
\drawloop[loopdiam=6,loopangle=90](A){{$\frac{1}{n}$}}
\drawloop[loopdiam=6,loopangle=90](B){{$\frac{2}{n}$}}
\drawloop[loopdiam=6,loopangle=90](C){{$\frac{3}{n}$}}
\drawloop[loopdiam=6,loopangle=90](D){{$\frac{n-2}{n}$}}
\drawloop[loopdiam=6,loopangle=90](E){{$\frac{n-1}{n}$}}
\drawloop[loopdiam=6,loopangle=90](F){1}
\end{picture}
\hrule
\caption{The stochastic process $X_t$ associated to Algorithm \ref{alg:ft-prob-gathering}}
\label{fig:stochastic-ftprob}
\end{figure}

As soon as a group of $[\frac{n}{2}]+1$ robots is formed, the convergence
needs only one additional round. That is, this group
will be the unique group of maximal cardinality. 
Therefore it will be an attractor for all the other robots and within
one additional round the convergence is achieved. 
In the following we compute the time needed to the stochastic 
process $(X_t)_{(t \in \NN^*)}$ to reach $[\frac{n}{2}]+1$. 
We define the expectation of the time needed for the stochastic process $X_t$ defined above
to reach state $k$, starting from state $l$. Formally,
\begin{eqnarray*}
T_l^k = \E \left[ min\{t \ \text{such that} \ X_t=k \ \text{knowing} \ X_0=l\} \right]
\end{eqnarray*}

According to Figure \ref{fig:stochastic-ftprob} we obtain the following induction formula:
\begin{eqnarray*}
T_l^{k}=1+\frac{l}{n}T_l^k+\frac{n-l}{n}T_{l+1}^{k}
\end{eqnarray*}
which leads to $(T_l^{k}-T_{l+1}^{k})=\frac{n}{n-l}$. 
Therefore, if we note $\alpha_n=[\frac{n}{2}]+1$
\begin{eqnarray*}
T_1^{\alpha_n }=\displaystyle \sum_{k=1}^{\alpha_n-1}
\frac{\alpha_n}{\alpha_n-k}= 
\displaystyle \alpha_n \sum_{k=1}^{\alpha_n-1}\frac{1}{k}  \leq \alpha_n ln(\alpha_n)
\end{eqnarray*}
\end{proof}

\subsection{Gathering in Fault-prone environments}
In this section we study the complexity of gathering in both crash and byzantine prone systems.
We prove  that crash tolerant gathering can be
achieved in $O(nln(n)+2f)$ while byzantine tolerant gathering is
exponential and needs additional assumptions (e.g. probabilistic scheduler).

\subsubsection{Crash-tolerant gathering}
In \cite{DGMP-disc06} we proved the impossibility of deterministic and
probabilistic weak gathering (gathering of correct robots) under centralized
bounded and fair schedulers and without additional assumptions.
An immediate consequence of this result is the necessity of 
additional assumptions (e.g.,~multiplicity knowledge), even for
probabilistic solutions under bounded schedulers. In the following we study the complexity of 
Algorithm \ref{alg:ft-prob-gathering} in fault-prone environments.

In order to compute the convergence of Algorithm
\ref{alg:ft-prob-gathering} we consider the worst scenario.
Recall that in a fault-free environment as soon a group of
$[\frac{n}{2}]+1$ robots is built, 
only one additional round is required to reach convergence. 
Our scenario goes as follows. Assume that as soon as the group of
maximal cardinality has $[\frac{n}{2}]+1$ robots, a crash occurs so
that the stochastic chain goes one step backward. We recall that only $f$ crashes may happen.
The following lemma computes the convergence time of Algorithm \ref{alg:ft-prob-gathering} 
according to the above described scenario.

\begin{lemma}
In a crash prone environment the convergence time of Algorithm
\ref{alg:ft-prob-gathering} 
is $O\left(\alpha_nln(\alpha_n)+ 2f \right)$ with $\alpha_n=[\frac{n}{2}]+1$.
\end{lemma}

\begin{proof}
We define  the following stochastic process to compute the convergence
time of the algorithm in a crash prone environment.
$\forall \ t, ~~ \ Y_t=k$ means that, at round $t$, the group with the
maximal cardinality has $k$ robots. The system transitions are as
follows:
\begin{itemize}
	\item $\PP[Y_t=k \mid Y_{t-1}=k]=\frac{k}{n}$
  \item $\PP[Y_t=k \mid Y_{t-1}=k-1]=\frac{n-k}{n}$
  \item Moreover, each time a crash occurs, the stochastic process
goes backward. 
\end{itemize}

According to Lemma \ref{lemma:ff-ct} the group of maximal cardinality
of $\frac{n}{2}$ robots is formed within $\alpha_nln(\alpha_n)$ rounds.
In the following we focus the time needed to the chain associated with
Algorithm \ref{alg:ft-prob-gathering} to move from state $\alpha_n-1$ to $\alpha_n$. The probability transition of this event is 
$\PP[Y_t=\alpha_n \mid Y_{t-1}=\alpha_n-1]= \frac{n-\alpha_n}{n}$. The mathematical expectation of the time needed
to perform this transition is : $\frac{n}{n-\alpha_n}\approx 2$. Therefore,
the convergence time of Algorithm \ref{alg:ft-prob-gathering} is $\alpha_nln(\alpha_n)+2f$.
\end{proof}

\begin{note}
Note that the above results hold even if the crashed robots are still physically present in the system but
stop the execution of any action.
\end{note}

\subsubsection{Byzantine-tolerant gathering}
In this section we address the complexity of byzantine-tolerant
probabilistic gathering. In the following $(n,f)$ denotes a system
where at most $f$ robots can have byzantine behavior.

In \cite{DGMP-disc06} we conjectured that  Algorithm~\ref{alg:ft-prob-gathering}
also solves $(n,f)$-weak byzantine-tolerant gathering problem when $n
\geq 3$ under bounded schedulers and multiplicity detection.
The counter-example shown in Figure \ref{fig:ce} advocates that the
boundedness of the scheduler should depend on the ration between correct
and byzantine robots otherwise the algorithm does not converge. That is,
assume an execution starting in a configuration with two groups of two
robots each and assume the right group has a byzantine robot (the
robot in red on the picture). The scheduler chooses one robot in the
left group and moves it in the right group. Then it chooses the
byzantine robot which (even if multiplicity is used) moves 
in the left group. Then the scheduler chooses any correct robot in the
right group and moves it in the left group. Then the byzantine is
chosen and moves to the right group. This configuration is symmetrical
to the initial configuration.

\begin{figure*}
\hrule
\hspace{3cm}\psfig{file=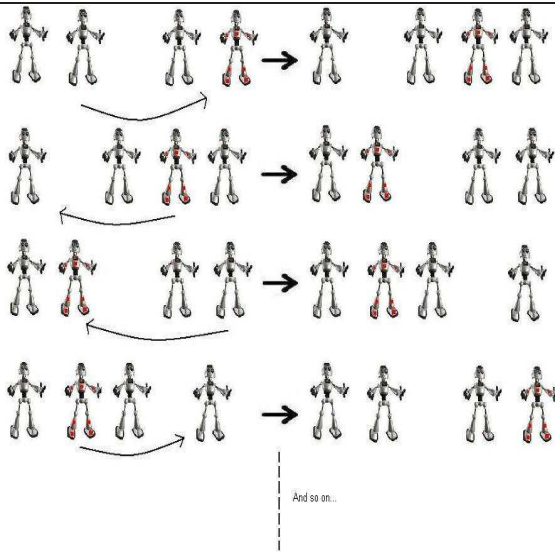, height=3in, width=4in}
\hrule
\caption{Counter Example: Algorithm~\ref{alg:ft-prob-gathering} does
not solve the problem under an arbitrary bounded scheduler}\label{fig:ce}
\end{figure*}

However byzantine-tolerant gathering is possible under special conditions.
The next lemma proves that byzantine-tolerant gathering is possible
when both the algorithm \textbf{and} the scheduler are
probabilistic. Here we use the power of random choice. That is, there is a
positive probability that the
byzantine robot is not chosen in a finite number of steps.

\begin{lemma}
  In systems with Byzantine faults, Algorithm \ref{alg:ft-prob-gathering}
  probabilistically solves the $(n,f)$-weak byzantine-tolerant gathering, $n \geq 3$, problem
  under a probabilistic scheduler and multiplicity detection.
  \end{lemma}

\begin{proof}
The proof  is based on the fact that as soon as there exists a group
of $\frac{N}{2}+1$ correct 
robots gathered, an additional round is needed to achieve convergence
 That is, every time a correct robot is selected by the scheduler, it
joins this group. 
Therefore, beyond this point,  the convergence time only depends on the scheduler.

In the following we study the probability to create a group verifying the
above stated property.
We define $\mathcal{L}$ : There exists a group of $\frac{N}{2}+1$ correct robots gathered.
\begin{eqnarray*}
\PP \left[ \ \text{reach} \ \mathcal{L} \ \text{in} \ \frac{N}{2}+1 \ \text{steps}  \right] > \displaystyle \left(\frac{1}{N}\right)^{(\frac{N}{2}+1)}= \varepsilon
\end{eqnarray*}
So
\begin{eqnarray*}
\PP \left[ \neg (\text{reach} \ \mathcal{L} \ \text{in} \ (\frac{N}{2}+1) \ \text{steps})  \right] \leq (1-\varepsilon)
\end{eqnarray*}
Therefore : 
\begin{eqnarray*}
\forall k, \ \PP \left[ \neg (\text{reach} \ \mathcal{L} \ \text{in} \ k(\frac{N}{2}+1) \ \text{steps})  \right] \leq (1-\varepsilon)^k
\end{eqnarray*}
\begin{eqnarray*}
\displaystyle \lim_{k \rightarrow \infty} \PP \left[ \neg (\text{reach} \ \mathcal{L} \ \text{in} \ k(\frac{N}{2}+1) \ \text{steps})  \right] = 0
\end{eqnarray*}
\end{proof}

Note that in our scenario the convergence time is exponential.
In order to simplify the calculations we consider: 
$\left(\frac{1}{N}\right)^N$ instead of $\left(\frac{1}{N}\right)^{(\frac{N}{2}+1)}$.

So, $\left[1-\left(\frac{1}{N}\right)^N \right]^t \leq \alpha$. This 
leads to: $$t ln\left(1-\left(\frac{1}{N}\right)^N \right)\leq
ln\alpha$$ 

Overall, $t$ verifies:
$$t\geq \frac{ln \alpha}{ln\left(1-\left(\frac{1}{N}\right)^N \right)} \sim \ln\left(\frac{1}{\alpha}\right)N^N$$

\begin{remark}
In order to prove the convergence we considered one of the worst
possible scenario. Therefore, we did not prove that the convergence time of the algorithm
is exponential. In order to do so, we would have to exhibit a set of
\textbf{non-null measure} 
of executions which converges in an exponential time.
\end{remark}

\section{Probabilistic scattering}
In this section we address another agreement problem : the
scattering. The unique existing probabilistic solution for scattering was proposed
in \cite{code} (see Algorithm \ref{SP}).
Analyzing the complexity of Algorithm \ref{SP} in both fault-free and fault-prone environments 
we prove that scattering is much easier to
achieve than gathering. The next section addresses Algorithm \ref{SP}
convergence in fault-free environments then we prove the correctness
and compute the complexity of the same algorithm in fault-prone systems.

\label{sec:scattering}
\subsection{Scattering in Fault-free environments}
Petit {\it et al.} \cite{code} proved that deterministic scattering is impossible
in ATOM model without additional assumptions and proposed a
probabilistic solution based on the use of Voronoi
diagrams defined below.
\begin{definition} 
Let ${\mathcal P}=\{p_1,p_2,\hdots,p_n\}$ be a set of points in the Cartesian 2-dimensional plane.
The Voronoi diagram of ${\mathcal P}$ is a subdivision of the plane into
$n$ cells, one 
for each point in ${\mathcal P}$. The cells have the property that a
point $q$ belongs to the Voronoi cell of point 
$p_i$ iff for any other point $p_j \in {\mathcal P}$,
$dist(q,p_i)<dist(q,p_j)$ where $dist(p,q)$ is the Euclidean distance
between $p$ and $q$. 
In particular, the strict inequality means that points located on the 
boundary of the Voronoi diagram do not belong to any Voronoi cell.
\end{definition}

The algorithm proposed in \cite{code} (see Algorithm \ref{SP}) is as follows.
Each robot uses a function $Random()$ that returns a value probabilistically 
chosen in the set $\{0,1\}$ : $0$ with probability $\frac{3}{4}$ and $1$ with
probability $\frac{1}{4}$. When a robot $r_i$ becomes active at time
$t$, it first computes the 
Voronoi Diagram of $P_{r_i}(t)$, i.e., the set of points occupied by the
robots. Then, $r_i$ moves toward a point
inside its Voronoi $Cell_i$ if $Random()$ returns $0$. 

\begin{algorithm}
    Compute the Voronoi Diagram;\\
    $Cell_i:=$ the Voronoi cell where $r_i$ is located; .\\
		$Current\_Pos:=$ position where $r_i$ is located;\\
    \noindent

    \textbf{If} $Random()=0$ \\
    \textbf{ then} Move toward an arbitrary position in $Cell_i$, different from $Current\_Pos$;\\
\caption{Probabilistic Scattering executed by robot $r_i$. \label{SP}}
\end{algorithm}

In the following we study the convergence time of Algorithm \ref{SP}.
\begin{lemma}
\label{lemma:scatter}
The convergence time of Algorithm~\ref{SP} is $O(n)$.
\end{lemma}

\begin{proof}
In order to study the convergence time of this algorithm we introduce the following stochastic process :
$X_t=k$ means that at time $t$ there are $k$ different Voronoi cells. 

Let us consider the worst scenario : If $X_t=k$, we assume that there
are $k-1$ robots at $k-1$ different positions and $n-k+1$ robots at the exact same position. 
Therefore, our stochastic process has the probability transitions given in Figure \ref{fig:scat}.

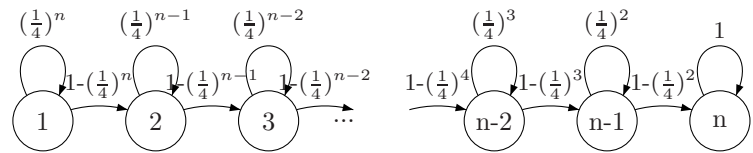
\begin{figure}[ht]
\hrule
\bigskip
\begin{picture}(130,25)
\node(A)(40,10){1}
\node(B)(55,10){2}
\node(C)(70,10){3}
\node(D)(100,10){n-2}
\node(E)(115,10){n-1}
\node(F)(130,10){n}
\node[Nframe=n](C')(85,10){}
\node[Nframe=n](D')(85,10){}
\node[Nframe=n](D'')(80,10){...}
\gasset{curvedepth=2}
\drawedge(A,B){\small1-$(\frac{1}{4})^n$}
\drawedge(B,C){\small1-$(\frac{1}{4})^{n-1}$}
\drawedge(C,C'){\small1-$(\frac{1}{4})^{n-2}$}
\drawedge(D',D){\small1-$(\frac{1}{4})^4$}
\drawedge(D,E){\small1-$(\frac{1}{4})^3$}
\drawedge(E,F){\small1-$(\frac{1}{4})^2$}
\drawloop[loopdiam=6,loopangle=90](A){\small$(\frac{1}{4})^n$}
\drawloop[loopdiam=6,loopangle=90](B){\small$(\frac{1}{4})^{n-1}$}
\drawloop[loopdiam=6,loopangle=90](C){\small$(\frac{1}{4})^{n-2}$}
\drawloop[loopdiam=6,loopangle=90](D){\small$(\frac{1}{4})^3$}
\drawloop[loopdiam=6,loopangle=90](E){\small$(\frac{1}{4})^2$}
\drawloop[loopdiam=6,loopangle=90](F){\small1}
\end{picture}
\hrule
\caption{The stochastic process $X_t$ associated to Algorithm \ref{SP}}
\label{fig:scat}
\end{figure}

For all $k \in \{1,\hdots,n \}$ let $T_l^k$ be the time for the stochastic process to enter in $k$, starting from $l$.
Based on the above notation and the chain proposed in Figure \ref{fig:scat}
we obtain the following induction formula~:
\begin{eqnarray*}
\forall k \in \NN,~ T_{l}^{k+1}=1+(1-\alpha^{n-(k-1)})T_l^k+\alpha^{n-(k-1)} T_{l}^{k+1}
\end{eqnarray*}
where $\alpha=\frac{1}{4}$.
This leads to $T_{l}^{k+1}-T_l^k=\frac{1}{1-\alpha^{n-(k-1)}}$.
If we sum from $0$ to $n-1$ we get :
\begin{eqnarray*}
T_1^n=\sum_{k=0}^{n-1}\frac{1}{1-\alpha^{n-(k-1)}}=\sum_{j=0}^{n-1}\frac{1}{1-\alpha^{j}}
\end{eqnarray*}
This sequence does not converge since $ \displaystyle  \lim_{j \rightarrow \infty}\frac{1}{1-\alpha^{j}} \neq 0$. Furthermore 
$\displaystyle  \frac{1}{1-\alpha^{j}} \mathop{\sim}_{j \rightarrow \infty} 1+\alpha^j$. Thus the partial sums are equivalent :
\begin{eqnarray*}
T_1^n \mathop{\sim}_{n \rightarrow \infty} \sum_{j=0}^{n-1}  1+\alpha^j
\end{eqnarray*}
Since $\sum_{j=0}^{n-1} (1+\alpha^j) = n + \sum_{j=0}^{n-1}  \alpha^j $
we finally obtain : $  T_1^n \displaystyle  \mathop{\sim}_{n \rightarrow \infty} n + \frac{4}{3}$
\end{proof}

\subsection{Scattering in fault-prone environments}
In this section we analyze the correctness and the complexity of
Algorithm \ref{SP} in crash-prone environments\footnote{Note that
\cite{code} does not address this issue}. Note that 
Algorithm \ref{SP} is not byzantine resilient.
In the sequel we denote by $(n,f)$ systems with $n$ robots 
where $f$ is the maximal number of
crashed robots. Note that scattering is impossible in systems where nodes may
crash since faulty nodes may share the same position. Therefore, we
consider in the following a weaker version of scattering, {\it weak
scattering} that requires the verification of the scattering
specification only from correct nodes.

\begin{lemma}
\label{lemma:scatft}
Algorithm \ref{SP} eventually verifies the $(n,f)$ weak scattering
specification under a weakly fair scheduler.
\end{lemma}

\begin{proof}
In the case when the $f$ faulty nodes disappear from the network the
convergence proof is similar to the convergence of the system when the number of robots
is $n-f$ (see \cite{code}).

Let consider systems where the faulty robots are still present in the system.
Let ${\mathcal M}$ be the set of correct robots whom position is occupied by
another robot (faulty or correct). In the following we show that
starting in a illegitimate configuration (${\mathcal M} \neq \emptyset$) the system converges
with positive probability in a finite number of steps 
to a configuration where ${\mathcal M}= \emptyset$. Let call the
latter configuration legitimate.
 
Consider an execution, $e$, starting in a illegitimate configuration, $c$
(${\mathcal M} \neq \emptyset$ in $c$). Assume the scheduler does not
choose robots in ${\mathcal M}$. Then, it may choose either faulty
robots or correct robots not in ${\mathcal M}$ (i.e. these robots 
do not share their position with any
other robot). In both cases, the size of ${\mathcal M}$ does not increase.
Since the scheduler is weakly fair it will eventually choose at least
one robot in ${\mathcal M}$.  Let $m \geq 1$ be the number of robots chosen by
the scheduler. Either all the $m$ robots choose randomly to stay or
leave the current position or some of them crash. In both cases the
size of ${\mathcal M}$ decreases by at least one robot.
With positive probability, $p \geq \frac{3}{4}(\frac{1}{4})^{m-1}$ this robot will
change its position to a new position in its Voronoi cell and hence
the size of ${\mathcal M}$ decreases by at least one. Recursively
repeating the same argument, in a finite number of steps, the size of
${\mathcal M}$ drops to 0 with positive probability.
\end{proof}

\begin{lemma}
\label{lemma:scatterft}
The convergence time of Algorithm \ref{SP} in systems with $n$ correct
robots but $f$ is $O(n)$.
\end{lemma}
\begin{proof}
The idea of the proof is similar to the one presented in Section \ref{sec:scattering}.
Consider a random variable $Y_t$ with values in $\{0$..$n\}$. 
$Y_t=k$ iff there is a set ${\cal M}$ of $k$ {\bf non faulty} robots
in which each robot shares its position with another robot. 
Our goal is to compute the time before
$Y_k$ reaches $0$. Note that the presence of faulty robots can only accelerate the
process (see the proof of Lemma \ref{lemma:scatft}). 
That is, if a robot of the set ${\cal M}$ crashes, 
our random variable is decreased by $1$.
Hence the convergence time is upper bounded by $n-f+\frac{4}{3}$. 
\end{proof}

\section{How to gather oblivious robots?}
\label{sec:gathering_via_scattering}
In Section \ref{sec:gathering} we analyzed a possible strategy to gather robots that converges in $O(nln(n))$ 
in fault-free environments and in $O(nln(n)+2f)$ in crash-proned systems. 
In Section \ref{sec:scattering} we computed the complexity of a scattering strategy 
that is one of optimal (i.e. $O(n)$ in fault-free environments and $O(n-f)$ in fault-prone environments). 
Interestingly, even if both strategies solve an agreement problem there 
is an important complexity gap between the existing implementations of gathering and scattering. 
An interesting open question is how to reduce this gap. That is, how to implement gathering in order 
to match the $O(n)$ lower bound?

In this section we discuss a promising alternative to obtain
gathering. The gathering algorithm proposed in \cite{code} 
is built on top of the  
probabilistic scattering algorithm shown in Figure \ref{SP} and 
works as follows: if there exist at least two positions 
with strict multiplicity then apply the scattering 
procedure otherwise apply any deterministic 
gathering protocol based on multiplicity knowledge. Note that 
the above protocol does not verify the gathering 
specification for the case when
the initial configuration does not contain strict multiplicity points.
The argument is similar to the one used in order to prove probabilistic
gathering impossible under an arbitrary scheduler (see \cite{DGMP-disc06}). 
That is, started in a configuration legitimate for  
scattering the above algorithm will try to apply 
the deterministic gathering. Since the gathering 
algorithm is deterministic several points of strict 
multiplicity may be created unless the scheduler 
is restricted to weaker forms (e.g. centralized or bounded). 
Then, the scattering restarts 
but the scheduler may ``derandomize'' the choice 
of the robots such that all the multiplicity points
are simultaneously destroyed.
From this point onward one can exhibit an infinite execution in which 
the system cycles between scattering 
and gathering without achieving convergence. However, the proposed 
method becomes interesting in systems where the flip/flop
scattering/gathering would be always able to converge to a configuration with an unique 
strict multiplicity point. Once this point created the system will need 
a single round to converge.  So, far no algorithm responds to these criteria.

\remove{ 
\begin{note}
Note that the lower bound to achieve gathering 
in systems with multiplicity knowledge and one 
single point of strict multiplicity is $n-2$. That is, 
assume an initial configuration with $2$ robots 
sharing the same position and all the other robots 
sharing different positions in the plane. In order 
to ensure the progress the scheduler will 
choose at least one robot at each configuration. 
Since the number of robots not gathered is initially 
$n-2$ then in at most $n-2$ 
steps the systems converges to a legitimate configuration.
\end{note}
}

\remove{
The following lemmas provide the complexity of constructive 
gathering in fault-free environments.

\begin{lemma}
The complexity of
constructive gathering in fault-free environments 
and systems with multiplicity knowledge started in configurations 
with a single is O(n).
\end{lemma}

\begin{proof}
In order to compute the costs of
constructive gathering we need to addition the cost 
of probabilistic scattering with the cost of 
deterministic gathering. Lemma \ref{lemma:scatter} shows 
that the scattering costs is $O(n)$ while 
Note \ref{note:detgathering} shows that gathering 
with multiplicity knowledge costs $O(n)$. Overall, the constructive gathering costs $O(n)$.
\end{proof}

Note that we do not consider the complexity of constructive 
gathering in crash-prone or byzantine-prone environments since 
in these environments scattering via this method is impossible even 
if the specification of the problem is relaxed to the weak gathering 
conform to the above discussion.
}

\section{Conclusions and Discussions}
\label{sec:conclusion}
The contribution of this paper is twofold. First, we proposed a
detailed complexity analysis of the existent probabilistic 
agreement algorithms (gathering and scattering) in both fault-free and
fault-prone environments. Moreover, using Markov chains tools and 
multiplicity knowledge we proved that the convergence time of 
gathering can be reduced from $O(n^2)$ (the best known tight 
bound) to $O(nln(n))$. Second, we prove that the best known scattering
bound is $O(n)$ (which is one to optimal). Additionally, we proved that in
crash-prone environments gathering is achieved in $O(nln(n)+2f)$
rounds while scattering needs $O(n+f)$ rounds.
Finally, we proposed a discussion related to the best strategy to design 
gathering. This work opens several research directions. 
For example the flip/flop scattering/gathering 
seems to be a promising direction to reduce the complexity of gathering. 
Another interesting direction is the complexity of byzantine-tolerant gathering. 
We conjecture that byzantine-tolerant gathering can be achieved in polynomial time. 
\bibliographystyle{plain}
\bibliography{gathering}
\end{document}